\documentclass{article}
\usepackage{amsmath,amsthm}
\newtheorem{theorem}{Theorem}
\newcommand{\yes}{Yes}
\newcommand{\no}{No}
\newcommand{\yesins}{\yes-instance}
\newcommand{\noins}{\no-instance}

\begin{document}
\title{Recognizing Linked Domain in Polynomial Time}
\author{Yongjie Yang\\
Saarland University}
\date{}
\maketitle

The celebrated Gibbard-Satterthwaite Theorem~\cite{Gibbard73,Satterthwaite75} states that when there are more than two candidates, any surjective social choice function which is defined over the universal domain of preferences and is strategy-proof must be dictatorial. Aswal, Chatterji and Sen generalize~\cite{Aswal2003} the Gibbard-Satterthwaite theorem by showing that the Gibbard-Satterthwaite theorem still holds if the universal domain is replaced with the linked domain.

An election is a tuple $E=(C,V)$ where $C$ is a set of candidates and $V$ is a multiset of votes each of which is defined as a linear order over $C$.
Two candidates $a,b\in C$ in $E$ are {\it{connected}} if there are two votes where one ranks $a$ first and $b$ second, and the other ranks $b$ first and $a$ second.
The election $E$ is {\it{linked}} if there is a linear order $(c_1,c_2,\dots,c_m)$ over $C$ such that the first two candidates in the order are connected and every other candidate is connected to at least two candidates before him in the order~\cite{Aswal2003}, i.e., $c_1$ and $c_2$ are connected, and every $c_i, i\geq 3$, is connected to at least two candidates in $\{c_1,c_2,\dots,c_{i-1}\}$.
In this note, we show that determining whether an election is linked is polynomial-time solvable.

\begin{theorem}
Determining whether an election is linked can be done in polynomial time.
\end{theorem}

\begin{proof}
Let $E=(C,V)$ be an election. Let $m=|C|$.
We create a graph $G=(V, A)$ as follows. For each candidate we create a vertex in $V$, and for each pair of connected candidates $a,b\in {C}$ we create an edge between $a$ and $b$. So, $|V|=m$. For a vertex $c$, $N(c)$ is the set of neighbors of $c$. Clearly, this graph can be constructed in polynomial time. The linked domain recognition problem is equivalent to determining whether there is a linear order $(c_{\pi(1)},\dots,c_{\pi(m)})$ of $V$ such that there is an edge between the first two vertices, and every $i$th vertex where $i\geq 3$ has at least two neighbors in the first $i-1$ vertices. We call this order a linked order of $G$. The algorithm guesses the first two vertices, say $a$ and $b$. The order of $a$ and $b$ does not matter. Let $c_{\pi(1)}=a, c_{\pi(2)}=b$ and $S_2=\{a,b\}$. This breaks down the instance into $\frac{m(m-1)}{2}$ subinstances. Then, we solve each subintance by a greedy algorithm. In particular, for each $i=3,4,\dots,m$, if there is a candidate $c\in C\setminus S_{i-1}$ such that $|N(c)\cap S_{i-1}|\geq 2$, then we order $c$ after the current candidate, i.e., let $c_{\pi(i)}=c$, and set $S_i=S_{i-1}\cup \{c\}$. If we find a linked order by doing so, then the subinstance is a {\yesins}. Otherwise, the subinstance is a {\noins}. The reason is as follows. Assume that when the algorithm is executed at some $i\in \{3,4,\dots,m\}$ every vertex of $C\setminus S_{i-1}$ has at most one neighbor in $S_{i-1}$. For the sake of contradiction, assume that there is a linked order $(c_{\pi'(1)},\dots,c_{\pi'(m)})$ such that $c_{\pi'(1)}=a$ and $c_{\pi'(2)}=b$. Let $j$ be the smallest integer such that $c_{\pi'(j)}\in C\setminus S_{i-1}$, i.e., $c_{\pi'(j)}$ is the first candidate among $C\setminus S_{i-1}$ in the linked order. Clearly, $j\geq 3$ and $\{c_{\pi'(1)},\dots, c_{\pi'(j-1)}\}\subseteq S_{i-1}$. However, according to our assumption, $c_{\pi'(j)}$ has at most one neighbor in $\{c_{\pi'(1)},\dots, c_{\pi'(j-1)}\}$, a contradiction.
\end{proof}

Recently, Pramanik~\cite{DBLP:journals/scw/Pramanik15} studied $\beta$-domain which differs from linked domain in that the connectedness condition is replaced with the weak connectedness condition. The above algorithm can be adapted to recognize $\beta$-domain in polynomial time as well. The only difference is that for the same election the auxiliary graph constructed in the algorithm may be (but not necessarily) different (in $\beta$-domain there is an edge between two candidates if they are weakly connected).

\bibliographystyle{plain}

\end{document}